\documentclass[11pt]{article}
\usepackage{amsfonts}
\usepackage{verbatim}
\usepackage{amssymb, amsmath}
\usepackage[margin=1.0in]{geometry}
\usepackage{graphicx}
\usepackage{algorithm}
\usepackage{algorithmic}
\usepackage{amsthm}
\newcommand{\Zp}{{\mathbb Z_p}}
\newcommand{\Ftwo}{{\mathbb F_2}}
\newcommand{\Reals}{{\mathbb R}}
\newcommand{\xor}{\oplus}

\newcommand{\GOOD}{{\mathcal{GOOD}}}

\newcommand{\lsb}{{\textnormal{lsb}}}

\newcommand{\Sample}{\textit{Sample}}
\newcommand{\Sign}{\textit{Sign}}
\newcommand{\SAMPLE}{\textit{SAMPLE}}

\newcommand{\drop}[1]{}

\newcommand{\floor}[1]{\lfloor {#1} \rfloor}

\newcommand{\Var}{{\sf Var}}
\newcommand{\E}{{\sf E}}

\newcommand{\req}[1]{(\ref{#1})}

\newtheorem {lemma} {Lemma}

\newtheorem {observation} [lemma] {Observation}

\newtheorem {theorem}[lemma] {Theorem}
\newtheorem {proposition}[lemma] {Proposition}

\newcommand\eps\varepsilon
\newcommand\fct\rightarrow

\newcommand\ceil[1]{\lceil {#1}\rceil}

\newcommand\ol\overline
\newcommand{\old}[1]{}

\begin{document}
%\title{A Memoryless Distributed Method to Maintain an MST with Low Communication}
%\title{A novel randomized tool replacing memory and communication requirements in MST and leader election algorithms}
\title{\texttt{Sample(x)=(a*x<=t)} is a distinguisher with probability 1/8}

\author{Mikkel Thorup\thanks{Partly supported by 
Advanced Grant DFF-0602-02499B from the Danish Council for Independent
Research under the Sapere Aude research career programme.
}\\
University of Copenhagen\\
Copenhagen, Denmark.\\
Email: \texttt{mikkel2thorup@gmail.com}}

\maketitle

\begin{abstract}
A random sampling function $\Sample:U\fct\{0,1\}$ for a key universe $U$
is a {\em distinguisher with probability $\alpha$\/} if for any given assignment
of values $v(x)$ to the keys $x\in U$, including at least one non-zero
$v(x)\neq0$, the sampled sum $\sum\{v(x)\mid x\in U \wedge \Sample(x)=1\}$
is non-zero with probability at least $\alpha$. Here the key values
may come from any commutative monoid (addition is commutative
and associative and zero is neutral). Such distinguishers were
introduced by Vazirani [PhD thesis 1986], and Naor and Naor used them
for their small bias probability spaces [STOC'90]. Constant
probability distinguishers are used for testing in contexts where the
key values are not computed directly, yet where the sum is easily
computed. A simple example is when we get a stream of key value pairs
$(x_1,v_1),(x_2,v_2),\ldots,(x_n,v_n)$ where the same key may appear
many times. The accumulated value of key $x$ is $v(x)=\sum\{v_i\mid
x_i=x\}$. For space reasons, we may not be able to maintain $v(x)$ for every key
$x$, but the sampled sum is easily maintained as the single value $\sum\{v_i\mid
\Sample(x_i)=1\}$.

Here we show that when dealing with $w$-bit integers, if $a$ is a
uniform odd $w$-bit integer and $t$ is a uniform $w$-bit integer, then
$\Sample(x)=[ax \mod 2^w\leq t]$ is a distinguisher with probability
1/8. Working with standard units, that is, $w=8, 16, 32,
64$, we exploit that $w$-bit multiplication works modulo $2^w$,
discarding overflow automatically, and then the sampling decision is
implemented by the C-code \texttt{a*x<=t}. Previous such samplers were
much less computer friendly, e.g., the distinguisher of Naor and Naor
[STOC'90] was more complicated and involved a 7-independent hash
function.\bigskip
\end{abstract}

\section{Introduction}
A random sampling function $\Sample:U\fct[2]=\{0,1\}$ for a set $U$ of
keys is viewed as sampling the keys $x\in U$ with $\Sample(x)=1$.  It is a {\em
  distinguisher with probability $\alpha$\/}, or, for short, an {\em
  $\alpha$-distinguisher}, if for any commutative monoid $R$ (addition is commutative
and associative and zero is neutral) and any given assignment of values
$v(x)\in R$ to the keys $x\in U$, including at least one non-zero $v(x)\neq
0$, the {\em sampled sum\/}
\[\sum\{v(x)\mid x\in U\wedge \Sample(x)=1\}\]
is non-zero with probability at least $\alpha$. Such distinguishers were
introduced by Vazirani \cite{u_vazirani86thesis}, and Naor and Naor used them
for their small bias probability spaces \cite{NaorN93}.  As an example, with
$R=\Ftwo$, the goal is to sample an odd number of elements from the 
set $X$ of elements $x\in U$ with $v(x)=1$. This is very
different from the more typical goal in sampling (see, e.g., \cite{CG89}) where
we with some small sampling probability $p$ hope to sample roughly $p|X|$
elements from $X$.

To simplify notation, we often identify the sampling function $\Sample$ 
with the set $\Sample^{-1}(1)$ of keys it samples from $U$, that is,
$x\in \Sample\iff \Sample(x)=1$. Thinking of $\Sample$ as a random subset of $U$
is often simpler, but our goal is to implement $\Sample$
as a membership function, returning $1$, or ``true'', for the elements in the subset.
Also, for any key set $X\subseteq U$,
we define the multiset $v(X)=\{v(x)\mid x\in X\}$. Then the sampled sum is written
as $\sum v(\Sample)$. Finally, we define a {\em non-zero value function\/}
$v:U\fct R$ as one that has a non-zero value $v(x)$ for at least one $x\in U$.

It is easy to see that a fully random sampling function $\Sample:U\fct[2]$ is 
a distinguisher with probability $1/2$; for consider
any non-zero value function $v:U\fct R$. We can make the independent sampling decisions $\Sample(x)$, $x\in U$,
one by one. When
we get to the last non-zero valued $x$ in $U$, if the sum over the previous
samples is non-zero, the final sampled sum is non-zero if $x$ is not
sampled, which happens with probability $1/2$. On the other hand, if
the sum over the previous
samples is zero, the final sampled sum is non-zero if $x$ is included,
which happens with probability $1/2$. Thus, if the
sample of $x$ is independent of the previous samples, then
we do get a distinguisher with probability $1/2$. The best possible
distinguisher samples a uniformly random non-empty set, that is,
we re-run the fully random sampling if no samples are made. It 
is a distinguisher with probability $1/2+1/2^u$ where $u=|U|$. Fully 
random sampling functions are, however, not realistic for
large sets of keys.

As a framework for more realistic hash functions, Wegman and Carter
\cite{wegman81kwise} define a random hash function
$h:U\fct[m]=\{0,\ldots,m-1\}$ to be {\em $k$-independent\/} if it is truly
random when restricted to any $k$ distinct keys $x_1,\ldots,x_k$, that
is, for any $t_1,\ldots,t_k\in[m]$, $\Pr[\bigwedge_{i=1}^k
  h(x_i)=t_i]=1/m^k$. 
Unfortunately, $k$-independent 
sampling does not in itself suffice for distinguishers unless we
have full independence with $k=u$. To see this, consider the constant 
value
function $v:U\fct\Ftwo$ where $v(x)=1$ for all $x\in U$. The
sampled sum over $\Ftwo$ is zero exactly if we sample an
even number of keys. Making a uniformly random such sample
is $(u-1)$-independent; for it can be described
by taking a fully random sample from any $u-1$ given keys, and
then make the sampling decision for the last key so that
the total number of samples is even. This
$(u-1)$-independent sampling is a distinguisher with probability $0$.

Thus $k$-independent sampling functions cannot be used directly as
distinguishers, but the concept is still useful in the design of
distinguishers. The previous constant probability distinguisher
defined by \cite[\S 3.1.1, pp. 843--846]{NaorN93} uses, as a subroutine,
a 7-independent hash function, and it is a distinguisher with probability
1/8.

In this paper we present some very simple and efficient constant
probability distinguishers.  First we have a generic construction
based on 2-independent hashing.  Let $h:U\fct[m]$ be a 2-independent hash
function. Moreover let the ``threshold'' $t$ be uniformly distributed
in $[m]$. Define $\Sample(x)=[h(x)\leq t]$. Assuming that the value
function $v:U\fct R$ has $n>0$ non-zeros, we will show that $\Pr[\sum
  v(\Sample)\neq 0]\geq (1-n^2/m^2)/8$, so for $n\ll m$, $\Sample$ is
essentially a $1/8$-probability distinguisher.  

For a canonical example of a 2-independent hash function $h$, we
pick a prime $p$ such that $U\subseteq \Zp$, set $m=p$, and choose $a$
and $b$ uniformly in $\Zp$. Then $h(x)=(ax+b)\mod p$ is 2-independent,
so $\Sample(x)=[(ax+b\bmod p)\leq t]$ is a distinguisher with
probability $(1-n^2/p^2)/8$. Not surprisingly, the addition of $b$
is superfluous. In fact, if $a$ is uniform in
$[m]_+=\{1,\ldots,m-1\}$, we will prove that $\Sample(x)=[ax\bmod p \leq
  t]$ is a strict$1/8$-probability distinguisher. More efficient
2-independent schemes exist, see, e.g., \cite{dietzfel96universal},
but we will show something even more efficient inspired by  the extremely
fast universal hashing scheme of Dietzfelbinger et
al.~\cite{dietzfel97closest}.  The new scheme works for $w$-bit integers, 
that is, with key set $U=[2^w]$. Let $a$ be a uniform odd $w$-bit integer and let $t$ be a
uniform $w$-bit integer. We will show that $\Sample(x)=[ax \mod
  2^w\leq t]$ is a distinguisher with probability $1/8$. The analysis
for this result is, technically, the hardest part of this
paper. Working with standard units, that is, $w=8, 16, 32, 64$, we can
exploit that $w$-bit multiplication works modulo $2^w$, discarding
overflow automatically, and then $\Sample(x)$ is implemented by the
C-code \texttt{a*x<=t}. We could also use \texttt{a*x<t}, but this yields a less clean analysis and slightly weaker bounds.

\subsection{Applications}\label{sec:appl}
We will now discuss some applications of distinguishers. These
applications are all known, and could be based on the distinguisher
from \cite{NaorN93}. Our contribution is
a very efficient distinguisher that is easy to implement, thus
making distinguishers more attractive to use in applications.

Constant probability distinguishers are used for testing in contexts where the
key values are not computed directly, yet where the sum is easily
computed. A simple example is when we get a stream of key value pairs
$(x_1,v_1),(x_2,v_2),\ldots,(x_n,v_n)$ where the same key may appear
many times. The accumulated value of key $x$ is $v(x)=\sum\{v_i\mid
x_i=x\}$. For space reasons, we may not be able to maintain $v(x)$ for every key
$x$, but the sampled sum is easily maintained as the single value $\sum\{v_i\mid
\Sample(x_i)=1\}$.

Another example suggested by \cite{NaorN93} is to implement Freivald's 
technique \cite{Freivalds77} to verify a matrix product $AB=C$
for $n\times n$ matrices. We use a random sampling function
$\Sample:[n]\fct [2]$, which we view as a $0/1$-vector 
$\vec s=(\Sample(0),\ldots,\Sample(n-1))^\top$.
We now just test if $A(B\vec s) = C \vec s$. This is trivially done
in $O(n^2)$ time, and if $\Sample$ is an $\alpha$-distinguisher, then
we discover $C\neq AB$ with probability $\alpha$. In Freivald's 
original algorithm \cite{Freivalds77}, the vector $\vec r$ is fully
random, but we generate it here using only $O(\log n)$ bits. We
note that for this concrete application, we have good specialized
techniques \cite{KS93}, but they do not allow us to
compute $\Sample(x)$ directly for a given key $x\in[n]$ in constant time.

In the streaming context above, testing is also a natural application:
if someone claims to have the correct value function $v':U\fct R$, 
then we can test if $v'-v$ is non-zero by comparing the sampled sums.

A quite different application from \cite{KKMT14}, inspired by \cite{AhnGM12,KapronKM13}, works with values
from $R=\mathbb F_2$. We have a tree $T$ in a graph $G=(V,E)$ and we want
to test if some edge from $E$ is leaving $T$.  The value of an edge is
the number of end-points it has in $T$, which in $\mathbb F_2$ is zero
if the edge has both end-points in $T$, so we have a non-zero value exactly
when there is an edge leaving $T$. Accumulating the sampled sum to
test for a non-zero is comparatively easy in many contexts.

Following \cite{NaorN93}, to decrease the probability of missing a
non-zero from $\alpha$ to $\eps$, we can apply $d=\ceil{\log_{\alpha}
  \eps}$ independent $\alpha$-distinguishers
$\Sample_1,\ldots, \Sample_d:U\fct[2]$, and ask if the vector of sampled
sums $\left(\sum v(\Sample_1),\ldots,\sum v(\Sample_d)\right)$ is all zeros. In
fact, these samplers do not really need to be independent. As
described in \cite[\S 3.2]{NaorN93}, based on a random walk in an
expander, we can generate $d=O(\log (1/\eps))$ random samplers using
$O(\log u +\log (1/\eps))$ random bits, so that the probability that all
sampled sums are zero is $\eps$.  We shall generally refer to this as
a {\em vector distinguisher}.

\drop{The above vector distinguisher $(\Sample_1(x),\ldots,\Sample_d(x))$ 
can be used to construct a $((1-\eps)/2)$-probability distinguisher. We combine it with a fully random bit-vector $\vec b=(b_1,\ldots,b_d)\in [2]^d$,
and define 
\[\SAMPLE(x)=\langle\vec b,(\Sample_1(x),\ldots,\Sample_d(x))\rangle=\sum_{i=1}^d
(\vec b_i \cdot \Sample_i(x)).\]
Since $\Pr[(\sum(v(\Sample_1)),\ldots,\sum v(\Sample_d))=0^d]=\eps$,
we get $(1-\eps)/2\leq \Pr[\SAMPLE(x)]\leq 1/2$.}

In \cite{NaorN93} the main focus is the case where the values are from
$\mathbb F_2$. The goal is to create an {\em $\eps$-biased sampling
  space}\footnote{We follow the definition from 
\texttt{http://en.wikipedia.org/wiki/Small-bias\_sample\_space}, which
is different but mathematically equivalent to the one in \cite{NaorN93}.}
defined by a sampling function $\SAMPLE:U\fct[2]$ that for any
non-empty subset $I\subseteq U$, samples an odd number of elements
from $I$ with probability $(1\pm\eps)/2$. Translating to our
framework, let the value function $v:U\fct \Ftwo$ be the
characteristic function for $I$, that is, $v(x)=1\iff x\in I$. Then
we want $\Pr[\sum v(\SAMPLE)=1]=(1\pm\eps)/2$. To define such a function $\SAMPLE$, we can use the above
vector distinguisher with $d=\lceil \log_{\alpha} \eps\rceil$ regular $\alpha$-probability
distinguishers $\Sample_1,\ldots, \Sample_d:U\fct[2]$, and take the inner product
with a fully random bit vector $\vec b=(b_1,\ldots,b_d)\in [2]^d$. Thus we define
\begin{equation}\label{eq:SAMPLE}
\SAMPLE(x)=\left\langle\vec b,\ (\Sample_1(x),\ldots,\Sample_d(x))\right \rangle=\sum_{i=1}^d
(b_i \cdot \Sample_i(x)).
\end{equation}
Since $\Pr[(\sum v(\Sample_1),\ldots,\sum v(\Sample_d))=0^d]\leq \eps$, we get
\[(1-\eps)/2\leq \Pr\left[\sum v(\SAMPLE)=1\right]\leq 1/2.\]
Like the vector distinguisher, this construction can be implemented
using $O(\log u+\log (1/\eps))$ random bits.

Naor and Naor's paper \cite{NaorN93} on small bias sampling spaces
have been very influential with more than 500 citations, and we cannot
cover all applications. The point in the current paper is to show that
constant probability distinguishers, which often play a central
role, now have a truly practical implementation. We are not
claiming any new asymptotic bounds in terms of $O$-notation.

\subsection{Comparison with previous implementations}
When it comes to the implementation of constant probability
distinguishers, we already mentioned the original construction of Naor
and Naor \cite[\S 3.1.1, pp. 843--846]{NaorN93} involving, among other
components, a 7-independent hash function, and coincidentally, it is
a 1/8-probability distinguisher like our \texttt{a*x<=t}. 

\paragraph{Small bias sampling spaces}
Alon et al.~\cite{AlonGHP92} have presented some simple constructions
of $\eps$-biased sampling spaces over $\Ftwo$. The goal is
to make sampling decisions $r_i$ for $i\in[n]$.
Quoting \cite{AlonGHP93}: {\em Our three constructions are so simple
that they can be described in the three corresponding paragraphs below:
\begin{enumerate}
\item A point in the first sample space is specified by two bit strings
  of length $m\stackrel{def}=\log (n/\eps)$ each, denoted $f_0\cdots f_{m-1}$ and
  $s_0\cdots s_{m-1}$, where $f_0=1$ and $f(t)=t^m+\sum_{i=0}^{m-1}
  f_i\cdot t^i$ is an irreducible polynomial. The $n$-bit sample
  string, denoted $r_0\cdots r_{n-1}$, is determined by $r_i=s_i$ for
  $i\leq m$ and $r_i=\sum_{j=0}^{m-1}f_j\cdot r_{i-m+j}$ for $i\geq m$.
\item A point in the second sample space is specified by residue $x$
  modulo a fixed prime $p\geq (n/\eps)^2$. The $n$-bit sample
  string, denoted $r_0\cdots r_{n-1}$, is determined by $r_i=0$ if $x+i$ is
  a quadratic residue modulo $p$ and $r_i=1$ otherwise.
\item A point in the third sample space is specified by two bit strings
  of length $m\stackrel{def}=\log (n/\eps)$ each, denoted $x$ and $y$. The $n$-bit sample
  string, denoted $r_0\cdots r_{n-1}$, is determined by letting $r_i$ equal the
inner product mod-2 of the binary vectors $x^i$ and $y$ where $x^i$ is the $i^{th}$ power of $x$ when considered as an element of $GF(2^m)$.
\end{enumerate}}
Mathematically, these descriptions are all very simple, and could all
be made to work if we want to generate all $n$ bits $r_i$. The setting
of this paper, however, is that of hashing, where
we want to compute $r_i$ directly for any given  key $i\in [n]$. Moreover,
we typically think of $n$ as much larger than the available space,
e.g., $n=u=2^{64}$ for 64-bit keys.
In this hashing context, Construction 1 has problems because
$r_i$ is determined from $r_{i-m},\ldots,r_{i-1}$.  Construction 2 and 3
both appear to require $\Omega(\log u)$ time for computing $r_i$
directly, e.g., using Euler's criterion for the quadratic residues. Thus, 
with $u\gg1/\eps$, it is
much better to use the function $SAMPLE$ from \req{eq:SAMPLE} which makes
$O(\log(1/\eps))$ constant time calls to a constant probability
distinguisher, e.g., either the one from \cite{NaorN93}, or, for even
better constants, the new \texttt{a*x<=t} distinguisher proposed in here.

We note that there has been interesting recent work
\cite{Ben-AroyaT13} on reducing the number of random bits needed to
generate $\eps$-biased sampling spaces. While the number of random
bits is important in some contexts, it not the concern in this paper.

\paragraph{AMS sketches}
In the case where the values are integers or reals, 
we can also get a distinguisher from the AMS sketches \cite{alon96frequency}.
These sketches are used to estimate 
$\sum_{x\in U} (v(x))^2$, which is non-zero if and only
if $v$ has a non-zero. Based on the analysis in  \cite{alon96frequency},
it is quite easy to see that if $v:U\fct\Reals$ has a non-zero,
then the AMS sketch returns non-zero with probability $1/3$.

The AMS sketch is based on 4-independent hashing, and using
the same analysis as in  \cite{alon96frequency}, we can easily show
that any $4$-independent sampling function $\Sample:U\fct[2]$ is
a distinguisher with probability $1/3$ when the values are reals.
This is a nice contrast to the situation with for $\Ftwo$ where
even $(|U|-1)$-independence does not suffice on its own.

We shall return to the AMS sketch and the 4-independent sampling of
reals in Section \ref{sec:AMS}.

\subsection{Theory and Practice}\label{sec:exper}
In theory of computing we often calculate running times with
$O$-notation, ignoring constant factors since we try
to identify the dominant term as the problem size $n$ goes to
infinity.  Nevertheless, constants do matter, and they can be appealing
to study when we know that we are getting close to the right answer,
e.g., for the number of comparisons needed for sorting.

In this paper, we show that a constant probability distinguisher can
be implemented as \texttt{a*x<=t} using just two standard instructions on integers. The
constant two is clearly optimal as no single standard instruction will
do on its own. 

As with many theoretical measures, we do not expect such an
instruction count to fully capture the practical performance on
real computers, because for example the relative performance of different
instructions vary from computer to computer. Nevertheless, with
such a small instruction count, we would expect it to be
very competitive in practice. 

\begin{table}
\begin{center}
\begin{tabular}{|c|r|}\hline
code/scheme & time (ns)\\\hline
\texttt{a*x<=t} & 1.02 \\
\texttt{a*x>>63} & 0.76 \\
\texttt{if (a*x<=t) S+=x;} & 1.15 \\
\texttt{if (a*x>>63) S+=x;} & 0.97 \\
7-independent hashing  & 52.12 \\\hline
\end{tabular}
\end{center}
\caption{Average time per computation over 10 million 
64-bit keys \textnormal{\texttt{x}} on a MacBook Air with a 1.7 GHz
Intel Core i5 Processor. The table compares our new \textnormal{\texttt{a*x<=t}} distinguisher
with the universal hashing scheme  \textnormal{\texttt{a*x>>63}} 
which is not a distinguisher but a general benchmark for fast hashing. They
are timed both alone, and inside a conditional statement.
The table also includes the average time for 7-independent hashing which is needed as a subroutine
of the previous distinguisher by \cite{NaorN93}}\label{tab:exper}
\end{table}

For 64-bit keys, we performed some simple experiments summarized in
Table \ref{tab:exper}. We compared our \texttt{a*x<=t} distinguisher
with the fastest known universal hashing scheme \texttt{a*x>>s} from
\cite{dietzfel97closest}. The time bounds reported are averaged over
10 million executions with different values of \texttt x.  More
precisely, we had a loop where we in each round incremented \texttt x
with the same 64-bit random number. Our reported time bounds include
the cost of this addition and the cost of the loop, but these costs
are much smaller than that of a multiplication.  With \texttt{s}
$=63$, both schemes can be viewed as returning a single bit, and both
combine a multiplication of $64$-bit integers with a single AC$^0$
operation, so a similar performance could be expected. Our
\texttt{a*x<=t} scheme took 1.02 ns whereas \texttt{a*x>>s} took 0.76
ns, so we were 34\% slower. A possible explanation is that while a
comparison itself is very fast, most computers do not return the
answer as a bit in a register.  Rather a comparison sets some flags
intended for later branching. To check this in the context of a
sampled sum, we tested ``\texttt{if (a*x<=t) S+=x;}'' which took 1.15
ns against ``\texttt{if (a*x>>63) S+=x;}'' which took 0.97 ns. Now the
difference is less than 20\%, and it is also smaller in absolute terms with
0.18 ns instead of 0.24 ns.  With this fast branching in mind, we
could implement the small bias sampler \SAMPLE{} from \req{eq:SAMPLE} using
logical {\sc and} (\texttt{\&\&}) rather than bit-wise {\sc and}
(\texttt{\&}), so as to benefit from short-circuit evaluation where
the evaluation of a Boolean expression terminates as soon as
the value is known. Thus we could implement \SAMPLE{} as 
\[\texttt{SAMPLE(x) \{B=0; for (i=0;i++;i<d) if (b[i]\&\&(a[i]*x<=t[i])) B++; return B\&1;\}}.\]
If \texttt{b[i]} is false, then with short-circuit evaluation, we skip immediately
to the next iteration of the loop.
In Section \ref{sec:non-distinguisher}, we show for completeness that
\texttt{a*key>>63} is not
a distinguisher. Rather we think of \texttt{a*key>>63} as a general
benchmark for fast hashing, and our new \texttt{a*x<=t} distinguisher 
was less than 50\% slower.

Recall now that the previous distinguisher from \cite[\S 3.1.1,
  pp. 843--846]{NaorN93} involves 7-independent hashing as a
subroutine, so it is at least as slow as 7-independent hashing. To
test 7-independent hashing, we used the implementation of
$k$-independent hashing of 64-bit keys from \cite[\S
  A.14]{thorup12kwise}. This is an tuned implementation in standard
portable C-code. It uses a random degree $d=k-1$ polynomial modulo
a prime $p>2^{64}$. As recommended in \cite{carter77universal}, it
uses a Mersenne prime $p=2^{89}-1$ so that ``$\bmod\;p$'' can be
computed with shifts and bit-wise Boolean operations. Still we have
the issue that standard 64-bit multiplication works modulo $2^{64}$, so
several of these are needed to implement a single multiplication modulo $2^{89}-1$.
As described in \cite[\S A.14]{thorup12kwise}, the implementation ends up
using $59d+7$ standard instructions, including $6d$
multiplications. We have $d=6$ for 7-independence, so we expect
7-independence to be much slower than our \texttt{a*x<=t} scheme.
In our experiments from Table \ref{tab:exper}, we see that
7-independent hashing takes 52 ns on the average. This is 50 times slower
than the 1.02 ns of our \texttt{a*x<=t} scheme.

The relative performance of different codes vary a lot from computer
to computer (see, e.g., the experiments in \cite[Table
  1]{thorup12kwise} on three different computers). There are also
other ways of implementing 7-independent hashing \cite{wegman81kwise},
particularly if one starts using more specialized instructions for
multiplication over certain fields \cite{IntelCGM}. One could also
have considered double tabulation \cite{Tho13:simple-simple}, but it
is only relevant for much higher independence, e.g., for 64-bit keys,
it requires 24 lookups in tables of size $2^{22}$, but
then it gives 100-independent hashing.

The above mentioned experiments only serve to illustrate
that our new scheme  \texttt{a*x<=t} is much faster
than that in \cite{NaorN93}. This may be important in time-critical situations like in the processing of a high volume data stream.
Perhaps more important for the potential of practical impact, our new 
 \texttt{a*x<=t} scheme is trivial to implement for anyone
considering the theory of distinguishers and small bias sampling
schemes for use in practice.

\subsection{Contents}
In Sections \ref{sec:2-indep}--\ref{sec:a*x} we will present our
general distinguishers, starting from the mathematically simplest
based on 2-independent hashing, and ending with the most efficient one
form the title. In the remaining sections, we will discuss some possible alternatives,
and their limitations.

\section{Using 2-independent hashing}\label{sec:2-indep}
First we present our simple generic distinguisher using
an arbitrary 2-independent
hash function $h:U\fct[m]$ and a uniformly random $t\in [m]$,
defining 
$\Sample:U\fct[2]$ as
\[\Sample^h_t(x)=[h(x)\leq t].\]
We also have a value function $v:U\fct R$ mapping keys $x\in U$ into elements
of some commutative monoid $R$ with plus as binary operator and zero the
neutral element. We assume that $v$ is non-zero for some key. 
Let $S\subseteq U$ be the set of keys with non-zero values, and let
$n=|S|>0$.
We define the {\em sampled sum\/}
\[V^h(t)=\sum v(\Sample^h_t)=\sum\{v(x)\mid x\in U\wedge \Sample^h_t(x)\}.\]
To prove that $\Sample^h_t$ is a distinguisher with probability $(1-n^2/m^2)/8$, we need to prove
\begin{equation}\label{eq:2-indep}
\Pr[V^h(t)\neq 0]\geq (1-n^2/m^2)/8.
\end{equation}
In most of our analysis, we will think of the hash function $h$ as
fixed.  Only at the very end will we use that $h$ is a random 2-independent
hash function. Let $x_{1},\ldots,x_{n}$ denote the keys in $S$
sorted according to their hash values, with ties broken
arbitrarily. For $i=1,\ldots,n-1$, let
$I_i=[h(x_{i}),h(x_{i+1}))$. We view $I_i$ as an interval of $h(x_{i+1})-h(x_{i})$ integers, which is empty if $h(x_{i})=h(x_{i+1})$.
Moreover, let $I_0=[0,h(x_{1}))$ and
    $I_n=[h(x_{n}),m)$. Note for any $I_i$, that $V^h(t)$ is a
      constant for $t\in I_i$, and we denote this constant
      $V^h(I_i)$. Then $V^h(I_0)=0$ and $V^h(I_i)=V^h(I_{i-1})+v(x_i)$ for
      $i=1,\ldots,n$. We say an interval $I_i$ is ``good'' if
      $V^h(I_i)\neq 0$; otherwise it is ``bad''. Let $\GOOD^h$ be
      the union of the good intervals.  Then $V^h(t)\neq 0$ if and
      only if $t\in \GOOD^h$. When $h$ is fixed while $t$ is uniformly
      random in $[m]$, we have $\Pr[V^h(t)\neq 0]=|\GOOD^h|/m$, so when
      $h:U\fct[m]$ is random,
\[\Pr_{h,t}[V^h(t)\neq 0]=\E_h[|\GOOD^h|]/m.\]
Generally, for each key $x_i\in S$, we define
the neighboring intervals $L^h(x_i)=I_{i-1}$ and $R^h(x_i)=I_i$. Since
$V^h(I_i)= V^h(I_{i-1})+v(x_i)$ for $i=1,\ldots,n$, we know that either
$L^h(x_i)$ or $R^h(x_i)$ is good. 

The analysis below is focused on an arbitrary $x\in S$. The length of
the good interval around $x$ is lower bounded by
$\ell^h(x)=\min\{|L^h(x)|,|R^h(x)|\}$. Therefore $\sum_{x\in S}\ell^h(x)/2$
is a lower bound on $|\GOOD^h|$. The following observation follows
directly from the definition of $\ell^h(x)$.
\begin{observation}\label{obs:cond-ell} $\ell^h(x)\geq \delta$ if and only if the
following conditions are satisfied:
\begin{itemize}
\item[(a)] for 
all $y\in S\setminus\{x\}$, $h(y)\not\in (h(x)-\delta,h(x)+\delta)$, and
\item[(b)] $\delta\leq h(x)\leq m-\delta$.
\end{itemize}
\end{observation}
\begin{lemma}\label{lem:prob-2-ind} If $h:U\fct [m]$ is 2-independent and
$|S|=n$, then 
$\Pr_h[\ell^h(x)\geq \delta]\geq 1-n(2\delta-1)/m$.
\end{lemma}
\begin{proof}
We are going to apply a union bound over the violating events from
Observation \ref{obs:cond-ell}.  For
each $y\in S\setminus\{x\}$ in (a), for any value of $h(x)$, there are
$2\delta-1$ violating values of $h(y)$, each happening with
probability $1/m$.  In (b) there
are $2\delta-1$ violating values of $h(x)$, each happening with
probability $1/m$. By a union bound, the probability that at least one of all
these violating events happen is bounded by $n(2\delta -1)/m$.
\end{proof}
\begin{theorem}\label{thm:2-ind} If $h:U\fct [m]$ is 2-independent
and $t$ is uniform in $[m]$, then $\Sample(x)=[h(x)\leq t]$ is
a $(1-n^2/m^2)/8$ probability distinguisher.
\end{theorem}
\begin{proof} We have $\E_h[\ell^h(x)]=\sum_{\delta=1}^\infty 
\Pr[\ell^h(x)\geq\delta]\geq\sum_{\delta=1}^k
\Pr[\ell^h(x)\geq\delta]$ for any integer $k$. Hence, by 
Lemma \ref{lem:prob-2-ind},
\begin{align*}
\E[\ell^h(x)]&\geq \sum_{\delta=1}^k\Pr[\ell^h(x)\geq\delta]\geq
\sum_{\delta=1}^{k} (1-(2\delta-1)n/m)\\ &=k-(2k(k+1)/2-k)n/m=k-k^2
n/m.
\end{align*}
The quadratic function $f:k\mapsto k-k^2n/m$ is concave and symmetric around
its maximum in the reals at $k_0=m/(2n)$. We pick the integer $k\in[k_0-1/2,k_0+1/2)$. Then
\begin{align*}
\E[\ell^h(x)]&\geq f(k)\geq f(k_0-1/2)\\
&= (m/(2n)-1/2)-(m/(2n)-1/2)^2n/m\\
&=(m/n-n/m)/4.
\end{align*}
Since $\sum_{x\in S}\ell^h(x)/2$ is a lower bound on $|\GOOD^h|$, we now get
\begin{align*}
\Pr_{h,t}[V^h(t)\neq 0]&=\E_h[|\GOOD^h|]/m\geq \E[\sum_{x\in S}\ell^h(x)]/(2m)\\
&\geq n/(2m)\cdot (m/n-n/m)/4=(1-n^2/m^2)/8.
\end{align*}
\end{proof}
\section{Simple tuning}\label{sec:tuning}
So far we have considered $h$ to be any 2-independent hash function, e.g.,
for some prime $p$, we could choose $a,b\in [p]$ uniformly at random,
and then use the classic $h(x)=(ax+b)\bmod p$. We will argue that it
suffices to use $h(x)=ax\bmod p$. Note that if $a\neq 0$, then $h$ is
a permutation. We will show that by picking $a$ uniformly in
$[p]_+=[p]\setminus\{0\}$, we get
a distinguisher with probability strictly bigger than $1/8$.

To analyze a hash function like $h(x)=ax\bmod p$, we need a better
lower bound on the good interval around a key. Consider the case where
$x=0$ is the only key with non-zero value. Then $h(x)=0$ so the
smallest interval around $x$ has length $\ell^h(x)=0$. However, recall
that $V^h(I_0)=0$ while $V^h(I_1)=v(x_1)\neq 0$. To exploit this, we
define a tighter lower bound $\ell^h_+(x)$ on the good interval around
$h(x)$. We set $\ell^h_+(x_1)=|I_1|$ while for $i=2,\ldots,n$,
$\ell^h_+(x_i)=\min\{|I_{i-1}|,|I_i|\}$. We can now tighten
Observation \ref{obs:cond-ell} for the new measure:
\begin{observation}\label{obs:cond-ell+} $\ell^h_+(x)\geq \delta$ if and only if the
following conditions are satisfied:
\begin{itemize}
\item[(a)] for 
all $y\in S\setminus\{x\}$, $|h(y)-h(x)|\geq \delta$, and
\item[(b)] $h(x)\leq m-\delta$.
\end{itemize}
\end{observation}
We note that Observation \ref{obs:cond-ell+} (a) is equivalent to
Observation \ref{obs:cond-ell} (a) while Observation
\ref{obs:cond-ell+} (b) has dropped the condition from Observation \ref{obs:cond-ell} (b) that $h(x)\geq
\delta$. 

The rest of the paper is devoted to hash functions of the form
\begin{equation}\label{eq:ax-m}
h_a^m(x)=ax\bmod m\textnormal,
\end{equation}
both with $m$ is a prime and with $m$ a power of two. We shall use
the notation
\[|x|_{\bmod m}=\min\{x \bmod m,-x \bmod m\}.\]
Also, for each $x\in S$, we define
\[S_x=\{y-x\mid y\in S, y\neq x\}\cup(\{x\}\setminus\{0\})\]
Note that $0\not\in S_x$. 
\begin{lemma}\label{lem:cond-ell+} $\ell^{h^m_a}_+(x)\geq \delta$ if 
for all $z\in S_x$, $|az|_{\bmod m}\geq\delta$.
\end{lemma}
\begin{proof}
The lemma follows easily from Observation \ref{obs:cond-ell+}.
First we note that $|h^m_a(y)-h^m_a(x)|<\delta$ implies
$|h^m_a(y)-h^m_a(x)|_{\bmod m}<\delta$. Moreover, $|h^m_a(y)-h^m_a(x)|_{\bmod m}=
|a(y-z)|_{\bmod m}$, so with $z=y-x$, our condition implies
Observation \ref{obs:cond-ell+} (a). Now, if $x=0$, we always
have $h^m_a(x)=0$, and then Observation \ref{obs:cond-ell+} (b) is
satisfied. Otherwise $x\in S_x$, but then $|az|_{\bmod m}\geq \delta$
implies that $h^m_a(x)\leq m-\delta$, which again implies Observation \ref{obs:cond-ell+} (b).
\end{proof}
In the rest of this section, $m=p$ is a prime and $a$ is uniform in $[p]_+$.
\begin{lemma}\label{lem:prob-ax-p} If $p$ is prime and $a$ is uniformly
distributed in $[p]_+=\{1,\ldots,p-1\}$, then for any $x\in S$,
$\Pr_h[\ell_+^{h_a^p}(x)\geq \delta]\geq 1-2n(\delta-1)/(p-1)$.
\end{lemma}
\begin{proof}
The proof is similar to that of Lemma  \ref{lem:prob-2-ind}.
From Lemma \ref{lem:cond-ell+} we know that 
$\ell_+^{h_a^p}(x)\geq \delta$ unless for some $z\in S_x$, we have 
$|az|_{\bmod p}<\delta$. Each such $z$ is non-zero and $a$ uniform in $[p]_+$, 
so $az\mod p$ is uniformly distributed in $[p]_+$. We therefore
get $|az|_{\bmod p}<\delta$ only if $az\mod p\in \{m-\delta+1,\ldots,m-1,1,\ldots,\delta-1\}$. This happens with probability $2(\delta-1)/(p-1)$. By a union bound, the probability that this happens for any of the at most $n$ elements in $S_x$
is at most $2n(\delta-1)/(p-1)$.
\end{proof}
\begin{theorem}\label{thm:ax-p} If $p$ is prime, $a$ is uniform in $[p]_+$,
and $t$ is uniform in $[p]$, then $\Sample(x)=[ax \bmod p \leq t]$ is
a $1/8$ probability distinguisher.
\end{theorem}
\begin{proof} The proof is similar to that of Theorem \ref{thm:2-ind},
except that we now use $\ell_+^{h_a^p}$ and Lemma \ref{lem:prob-ax-p}.
For any given integer $k$, we get
\begin{align*}
\E[\ell_+^{h_a^p}(x)]&\geq \sum_{\delta=1}^k\Pr[\ell_+^{h_a^p}(x)\geq\delta]\geq
\sum_{\delta=1}^{k} (1-2n(\delta-1)/(p-1))\\
&=k-k(k-1)n/(p-1)= k(1+n/(p-1))-k^2n/(p-1)).
\end{align*}
The quadratic function $f:k\mapsto k(1+n/(p-1))-k^2n/(p-1))$ is concave and symmetric around
its maximum in the reals at $k_0=(p-1)/(2n)+1/2$. 
We pick the integer $k\in[k_0-1/2,k_0+1/2)$. Then
\begin{align*}
\E[\ell_+^{h_a^p}(x)]&\geq f(k)\geq f(k_0-1/2)\\
&=(p-1)/(2n)\cdot(1+n/(p-1))-((p-1)/(2n))^2n/(p-1))\\
&=(p-1)/(2n)+1/2-(p-1)/(4n)=p/(4n)-1/(4n)+1/2>p/(4n).
\end{align*}
The probability that we get distinction with a non-zero sampled sum 
$V^{h_a^p}(t)$ is thus lower bounded by
\begin{align*}
\Pr_{h,t}[V^{h_a^p}(t)\neq 0]&=\E_h[|\GOOD^{h_a^p}|]/p
\geq \E[\sum_{x\in S}\ell_+^{h_a^p}(x)]/(2p)
>n/(2p)\cdot p/(4n)=1/8.
\end{align*}
\end{proof}

\section{\texttt{Sample(x)=(a*x<=t)}}\label{sec:a*x}
We will now study a very computer friendly implementation of
$h_a^m(x)=ax\bmod m$ from equation \req{eq:ax-m}. For some integer parameter
$w$, we set $m=2^w$ and pick $a$ as a uniformly distributed odd number
in $[m]$. Since $a$ is odd, it is relatively prime to $m$ which is a
power of two, so $h_a^{2^w}$ is a permutation on $[m]$.  If $w$ represents the bit length of a standard unit, that is, if $w\in
\{8,16,32,64\}$, then we get a particularly efficient portable code in
a programming language such as C. When $t$ is also a uniform $w$-bit
number, we can the compute the sampling decision for $x$ with the full
code \texttt{a*x<=t}. An important point here is that $w$-bit
multiplication automatically discards overflow, so \texttt{a*x}
computes $ax\bmod 2^w$. We will show that this power-of-two scheme
yields a distinguisher with probability bigger than $1/8$. 

We note that the above power-of-two scheme with odd $a$
has been studied before by Dietzfelbinger et al.~\cite{dietzfel97closest}. They proved
that the scheme $h_a^{m,s}(x)=\floor{(ax\bmod 2^w)/2^s}$ mapping $[2^w]$ into
$[2^{w-s}]$ is universal in the sense that for any distinct $x,y\in [2^w]$,
$\Pr[h(x)=h(y)]\leq 2^{s-w}$. This scheme is implemented efficiently 
as \texttt{(a*x)>>s}. Our analysis here is, however, quite different from
that of Dietzfelbinger et al.~\cite{dietzfel97closest}.

By Lemma \ref{lem:cond-ell+}, for any $k$ we choose, we get
\begin{align}
\E[\ell_+^{h_a^{2^w}}(x)]&\geq \sum_{\delta=1}^k\Pr[\ell_+^{h_a^{2^w}}(x)\geq\delta]\geq
\sum_{\delta=1}^{k} (1-\sum_{z\in
  S_x}\Pr[|az|_{\bmod 2^w}<\delta])\nonumber\\ 
&=k-\sum_{z\in
  S_x}\sum_{\delta=1}^{k} \Pr[|az|_{\bmod 2^w}<\delta].\label{eq:sum}
\end{align}
In the previous section, with $m$ prime and random $a\neq [m]_+$, we
proved that $\Pr[|az|_{\bmod m}<\delta]\leq2(\delta-1)/(p-1)$ for
any $z\in[m]_+$ (cf.~the proof of Lemma \ref{lem:prob-ax-p}). However,
for our power-of-two scheme with $m=2^w$, we do not have a good
general bound on $\Pr[|az|_{\bmod m}<\delta]$ since the probability may change
a lot with $z$. It turns out, however, that by exchanging summation order, 
we can get a good bound on the sum 
$\sum_{\delta=1}^{k} \Pr[|az|_{\bmod 2^w}<\delta]$.
\begin{lemma}\label{lem:good-sum} For any given $k,z\in [2^w]_+$, and uniformly random
odd $a\in [2^w]$,
\[\sum_{\delta=1}^{k} \Pr[|az|_{\bmod 2^w}<\delta] \le 2^{2-w}\lfloor k/2 \rfloor \lceil k/2 \rceil.\]
If $k$ is even, the bound evaluates to $k^2/2^w$, and if $k$ is odd, it is 
$(k-1)(k+1)/2^w=(k^2-1)/2^w$.
\end{lemma}
\begin{proof}
Let $k,z\in [2^w]_+$ be fixed, and let $i=\lsb(z)$ denote the position of the least significant set bit in
the binary representation of $z$, that is, $z=z'2^i$ for some odd
$z'$.  The multiplier $a$ is
a uniformly distributed odd $w$-bit number $a$, and this
means that $az\bmod 2^w$ is a uniformly distributed odd multiple of $2^i$, that
is, a number of the form $2^i+j2^{i+1}$ for $j\in [2^{w-i-1}]$.
There are the same number of such odd multiples in $(2^w-\delta)$ and
in $(0,\delta)$, and none in $0$, so we conclude that
\[\Pr[|az|_{\bmod 2^w}<\delta]=
2\cdot\sum_{j\,\colon\,0 < 2^i + j2^{i+1} < \delta}1/ 2^{w-i-1} = 2^{2-w}\cdot\sum_{j\,\colon\,0 < 2^i + j2^{i+1} < \delta} 2^i.
\]
Hence
\begin{align*}
2^{w-2}\cdot \sum_{\delta=1}^{k} \Pr[|az|_{\bmod 2^w}<\delta]
&= \sum_{1\le \delta \le k} \;\; \sum_{j\colon 0 < 2^i + j2^{i+1} < \delta} 2^i\\
&= \sum_{0 \le j < (k-2^i)/2^{i+1}} \;\; \sum_{\delta = 2^i + j2^{i+1}+1}^k 2^i\\
&= \sum_{0 \le j < (k-2^i)/2^{i+1}}  (k - 2^i - j2^{i+1})2^i\\
&= \sum_{0\le j < J_k} (k - 2^i - j2^{i+1})2^i,
\end{align*}
where $J_k=  \lceil(k-2^i)/2^{i+1}\rceil$. Hence
$$
2^{w-2}\cdot \sum_{\delta=1}^{k} \Pr[|az|_{\bmod 2^w}<\delta]
= J_k2^i(k - 2^i) - 2^{2i+1}J_k(J_k-1)/2 = J_k2^i(k - J_k2^i).
$$ 
The quadratic mapping $\alpha\mapsto \alpha(k-\alpha)$ is concave and 
symmetric around its maximum in the reals at $\alpha_0=k/2$. The
maximum in the integers is attained at the points $\lfloor k/2
\rfloor$ and $\lceil k/2 \rceil$ where it evaluates to $\lfloor k/2
\rfloor(k-\lfloor k/2 \rfloor)= \lfloor k/2 \rfloor\lceil k/2
\rceil$. We have $\alpha=J_k2^i=\lceil(k-2^i)/2^{i+1}\rceil 2^i$, which is indeed
an integer, and equal to $\lceil k/2 \rceil$ when $i=0$. We conclude
that
\[
2^{w-2}\cdot \sum_{\delta=1}^{k} \Pr[|az|_{\bmod 2^w}<\delta]
= J_k2^i(k - J_k2^i)\leq \lfloor k/2 \rfloor\lceil k/2
\rceil.\]
This completes the proof of Lemma \ref{lem:good-sum}.
\end{proof}

\begin{lemma}\label{lem:ell-2^w} For any $x\in S$, when $a$ is a uniform odd number in
$[2^w]$, we have
\[\E[\ell_+^{h^{2^w}_a}(x)]\geq 2^w/(4n).\]
\end{lemma}
\begin{proof}
Inserting the bound of Lemma \ref{lem:good-sum} in \req{eq:sum}, for any odd positive integer $k$, 
\[
\E[\ell_+^{h_a^{2^w}}(x)]=k-\sum_{z\in
  S_x}\sum_{\delta=1}^{k} \Pr[|az|_{\bmod 2^w}<\delta] \ge k - n \cdot (k^2-1)/{2^w}.\]
The quadratic function $f:k\mapsto k - n \cdot (k^2-1)/{2^w}$ is concave and 
symmetric around its maximum in the reals at $k_0 = 2^{w-1}/n$.
We pick $k$ as the odd integer in $[k_0-1,k_0+1)$.
Then 
\begin{align*}
\E[\ell_+^{h_a^{2^w}}(x)]&=f(k)\ge f(k_0-1) = (k_0-1) - n\cdot k_0(k_0-2)/2^w\\
&=\frac{2^{w-1}}{n}-1 - \frac{n}{2^w}\cdot\frac{2^{2(w-1)}}{n^2} + \frac{n}{2^w}\cdot \frac{2^w}{n}\\
&=\frac{2^{w}}{2n}-1 - \frac{2^w}{4n} + 1=\frac{2^w}{4n}.
\end{align*}
\end{proof}

\begin{theorem}\label{thm:ax-2^w} If $w$ is a natural number, $a$ is a uniform
odd number in $[2^w]$, and $t$ is uniform in $[2^w]$,
then $\Sample(x)=[ax \bmod 2^w \leq t]$ is
a distinguisher with probability $1/8$.
\end{theorem}
\begin{proof}
Recall that $\ell^{h^{2^w}_a}_+\!(x)$ is a lower bound on the size of
the good interval neighboring $h^{2^w}_a(x)$, so $\sum_{x\in S}
\ell^{h^{2^w}_a}_+\!(x)/2$ is a lower bound of the total length
$|\GOOD^{h_a^{2^w}}|$ of the good intervals. Hence, by Lemma
\ref{lem:ell-2^w},
\[\E_a[|\GOOD^{h_a^{2^w}}|]
\geq \sum_{x\in S}\E[\ell_+^{h_a^{2^w}}(x)]/2
>n (2^w/(4n))/2=2^w/8.\]
We get a distinguisher when the sampled sum $V^{h_a^{2^w}}(t)$ is non-zero. By definition, this is exactly when 
$t\in \GOOD^{h_a^{2^w}}$, which for a given $\GOOD^{h_a^{2^w}}$ happens
with probability $|\GOOD^{h_a^{2^w}}|/2^w$. Therefore
\begin{align*}
\Pr_{a,t}[V^{h_a^{2^w}}(t)\neq 0]&=\E_h[|\GOOD^{h_a^{2^w}}|]/2^w
\geq (2^w/8)/2^w=1/8.
\end{align*}
\end{proof}
Recall that when $w\in \{8,16,32,64\}$, we implement $[ax \bmod 2^w \leq t]$ in C
as \texttt{(a*x<=t)}. Theorem~\ref{thm:ax-2^w} thus
formalizes the title of this paper.

\section{Relation to AMS sketches and 4-independence with real values}\label{sec:AMS}
When we have a value function $v:U\fct\Reals$, we can
use the classic AMS sketches \cite{alon96frequency} as
a kind of distinguisher. These sketches are used to estimate the
second moment $Q=\sum_{x\in U} (v(x))^2$, which is non-zero if and only
if $v$ has a non-zero.

The AMS sketch takes
a random sign function $\Sign:U\fct \{-1,+1\}$ and 
computes the sketch
\[X=\sum_{x\in U} v(x)\Sign(x).\]
Then $Y=X^2$ is used as an estimator for $Q$. 
In \cite{alon96frequency} it is proved that $\E[X^2]=Q$ if $\Sign$ is 2-independent and that
 $\Var[X^2]<2Q$ when $Q>0$ and $\Sign$ is 4-independent. From the simple
observation below it follows that $Y$ and hence $X$ are non-zero
with probability bigger than $1/3$.
\begin{observation}\label{obs:non-zero} If $Y\geq 0$ is a random variable and $\Pr[Y=0]=p$,
then $\Var[Y]\geq \frac{p}{1-p}\E[Y]^2$. In particular, if $p\geq 2/3$
then $\Var[Y]\geq 2\E[Y]^2$.
\end{observation}
\begin{proof} The observation follows from the more general Paley-Zygmund inequality, but
for this simple case, we provide a direct proof.  Given that $\Pr[Y=0]=p$, we get the smallest
variance for a given mean $\gamma$ if $Y$ is a constant $y$ when $Y\neq 0$.
Then $\Pr[Y=y]=(1-p)$ and $\mu=(1-p)y$, so $\Var[Y]=\E[Y^2]-\gamma^2=(1-p)(\gamma/(1-p))^2-\gamma^2=\frac p{1-p}\gamma^2$.
\end{proof}
We now observe that we get the same result if we replace the 
4-independent sign-function $\Sign:U\fct\{-1,+1\}$ with a 
4-independent sample function $\Sample:U\fct[2]$, that is,
\begin{proposition}\label{prop:4-indep}
A 4-independent $\Sample:U\fct[2]$ is a distinguisher with probability $1/3$
for real valued functions. 
\end{proposition}
\begin{proof}
Consider a non-zero value function $v:U\fct \Reals$. 
Let $x_1,\ldots,x_n\in\Reals\setminus\{0\}$ be the non-zero
values in the multiset $v(U)$. Each $x_i$ gets sampled
with probability $1/2$, yielding a random variable $X_i$ with
$\E[X_i]=\mu_i=x_i/2$. Our sampled sum is $X=\sum_i X_i$ which
has mean $\mu_i=\sum_i \mu_i$. We want to argue that
$X\neq 0$ with constant probability using that the $X_i$ are 4-independent.

We are going to study the equivalent event that $X^2\neq 0$. 
We will argue that 
\begin{equation}\label{eq:var}
\Var[X^2]<2\E[X^2]^2
\end{equation}
By Observation \ref{obs:non-zero} this implies that $\Pr[X^2=0]<2/3$,
hence that  $\Pr[X\neq 0]>1/3$, as desired. 

Since $\Var[X^2]=\E[X^4]-\E[X^2]^2$, we have the following
equivalent form of \req{eq:var}:
\begin{equation}\label{eq:fourth}
\E[X^4]<3\E[X^2]^2.
\end{equation}
To prove \req{eq:fourth}, we will compute $\E[(X-\mu)^k]$ and $\E[X^k]$ for 
$k=2,3,4$. 
The calculations are straightforward because of the symmetry
$\Pr[(X_i-\mu_i)=\mu_i]=\Pr[(X_i-\mu_i)=-\mu_i]=1/2$. Therefore 
$E[(X_i-\mu_i)^k]=0$ for $k$ odd 
while $E[(X_i-\mu_i)^k]=\mu_i^k$ for $k$ even. Thus
\begin{align*}
\E[(X-\mu)^2]&=\E[X^2]-\mu^2\\
\E[(X-\mu)^2]&=\E[(\sum_i X_i-\mu_i)^2]=\sum_i \E[(X_i-\mu_i)^2]=\sum_i \mu_i^2\\
\E[X^2]&=\sum_i \mu_i^2 +\mu^2.\\
\end{align*}
\begin{align*}
\E[(X-\mu)^3]&=\E[X^3]-3\E[X^2]\mu+2\mu^3\\
\E[(X-\mu)^3]&=\E[(\sum_i X_i-\mu_i)^3]=0\\
\E[X^3]&=3\E[X^2]\mu-2\mu^3
=3(\sum_i \mu_i^2 +\mu^2)\mu-2\mu^3\\
&=3\mu\sum_i \mu_i^2 +\mu^3\\
\end{align*}
\begin{align*}
\E[(X-\mu)^4]&=\E[X^4]-4\E[X^3]\mu+6\E[X^2]\mu^2-3\mu^4\\
\E[(X-\mu)^4]&=\E[(\sum_i X_i-\mu_i)^4]=\E[\sum_i (X_i-\mu_i)^4]
+3\E[\sum_{i\neq j} (X_i-\mu_i)^2(X_j-\mu_j)^2]\\
&=\sum_i \mu_i^4+3\sum_{i\neq j} \mu_i^2\mu_j^2
=3(\sum_i \mu_i^2)^2-2\sum_i \mu_i^4.
\end{align*}
\begin{align*}
\E[X^4]&=\E[(X-\mu)^4]+4\E[X^3]\mu-6\E[X^2]\mu^2+3\mu^4\\
&=3(\sum_i \mu_i^2)^2-2\sum_i \mu_i^4+4(3\mu\sum_i \mu_i^2 +\mu^3)\mu-
6(\sum_i \mu_i^2 +\mu^2)\mu^2+3\mu^4\\
&=3(\sum_i \mu_i^2)^2-2\sum_i \mu_i^4+
12\mu^2\sum_i \mu_i^2+
4\mu^4-
6\mu^2\sum_i \mu_i^2 -6\mu^4+3\mu^4\\
&=3(\sum_i \mu_i^2)^2-2\sum_i \mu_i^4+
6\mu^2\sum_i \mu_i^2+\mu^4\\
&=2\E[X^2]^2-2\sum_i \mu_i^4-2\mu^4\\
&<2\E[X^2]^2.
\end{align*}
The last strict inequality follows because $v$ was non-zero, hence that
we have some $\mu_i\neq 0$.
This completes the proof of \req{eq:fourth}, hence of \req{eq:var}.
\end{proof}
Proposition \ref{prop:4-indep} forms an interesting
contrast to the fact that less than full independence does not suffice for
a distinguisher for values in $\Ftwo$. 

Contrasting the above result, we will now argue that
2-independence does not suffice for a distinguisher for the reals.
\begin{proposition} \label{prop:2-not}
There is a 2-independent sampling function $\Sample:[u]\fct[2]$ and a 
non-zero value function $v:[u]\fct \Reals$ such 
that $\Pr[\sum v(\Sample)\neq 1]\leq 4/u$.
\end{proposition}
\begin{proof}
First we describe the value function. We assume $u=4n$ for
some positive integer $n$. We have $2n$ values that are
$+1$ and $2n$ values that are $-1$. The positive and negative
values will be sampled independently, but using the same
scheme. In both
cases, for some parameter $\eps$, the distribution samples nothing 
with probability $\eps$ and all with probability $\eps$.
Otherwise make a ``balanced'' sample of exactly $n$ random values. 
If we get the balanced case for both negative and positive
values, then the sampled sum is zero, and this happens
with probability at least $1-4\eps$.

We will now fix $\eps$ so as to get a $2$-independent distribution.
Regardless of $\eps$, in the marginal distribution, each item 
is sampled with probability $1/2$. Because everything is symmetric, it now suffices
to set $\eps$ such that any 2 given items are sampled with
probability $1/4$. Thus we want
\begin{align*}
&1/4=\eps+(1-2\eps)(1/2)(n-1)/(2n-1)=\frac{\eps(4n-2)+n-1-2\eps(n-1)}{4n-2}=
\frac{n-1+2\eps n}{4n-2}\\
&\iff \eps=1/(4n)=1/u.
\end{align*}
The probability of getting a non-zero sum is thus at most $4/u$.
\end{proof}

\section{Simple tabulation}
We will now briefly discuss what can be achieved using simple
tabulation for constructing distinguishers. Simple tabulation dates
back to Zobrist \cite{zobrist70hashing}. We view a key $x\in [u] =
\{0,\ldots,u-1\}$ as a vector of $c>1$ characters
$x_0,\ldots,x_{c-1}\in \Sigma = [u^{1/c}]$. For each character
position, we initialize a fully random table $H_i$ of size $|\Sigma|$
with values from $R = [2^r]$. The hash value of a key $x$ is
calculated as
\[
    h(x) = H_0[x_0]\xor\cdots\xor H_{c-1}[x_{c-1}]\ .
\]
P\v{a}tra\c{s}cu and Thorup \cite{patrascu12charhash} analyzed simple
tabulation, showing that it works well in many algorithmic contexts.
It is now tempting to set $r=1$ and use $\Sample=h:[u]\fct[2]$ as a sampler.

We will now argue that the above sampler is not a distinguisher for
values in $\Ftwo$. Suppose $c=2$ and consider the 4 keys $aa$, $ab$,
$ba$, and $bb$. Then over $\Ftwo$, we have $\Sample(aa)+
\Sample(ab)+\Sample(ba)+\Sample(bb)=0$.  This is indeed the
standard proof that simple tabulation is not 4-independent. It
implies that if our four keys are the only ones with value $1$,
then the sampled sum is always $0$.

On the other hand, we do claim that our simple tabulation sampler
is a constant probability distinguisher over the reals. In Proposition
\ref{prop:4-indep} we proved this to be the case when $\Sample$ is
4-independent, but simple tabulation is not 4-independent.
However, in \cite{patrascu12charhash} it is proved
that the fourth moment that we studied in the proof of Proposition
\ref{prop:4-indep} deviates only by a constant
factor from that with full randomness, and all lower moments
are the same since simple tabulation is 3-independent. We can therefore
perform the same analysis as that in the proof of Proposition \ref{prop:4-indep}, and conclude that simple tabulation sampling is a constant
probability distinguisher for real values.

In the experiments with a 64-bit computer from \cite[Table
  1]{patrascu12charhash}, the universal hashing \texttt{a*x>>s} takes
2.33 ns whereas simple tabulation takes 11.40 ns. In our own
experiments, our distinguisher \texttt{a*x<=t} was only 34\% slower
than \texttt{a*x>>s}, so it should be 2-3 times faster than simple
tabulation (the relative  difference may vary a lot from computer to computer 
since \texttt{a*x<=t} depends on arithmetic operations while tabulation depends
on the speed of cache). On top of that, simple tabulation uses much more space and
it is not a general distinguisher.  However,
simple tabulation does have one interesting advantage as a
distinguisher for reals; for if we use $r>1$ output bits in $h$ and
the tables $H_i$, then they are completely independent, yielding $r$
independent distinguishers that can be used for the vector
distinguisher mentioned in Section \ref{sec:appl}. Unfortunately,
simple tabulation is not a general distinguisher, and we do not
know of a better way of generating bits for a general vector distinguisher
than to generate each one individually using our new \texttt{a*x<=t}
scheme. This is left as an interesting open problem.

\section{Other non-distinguishers}\label{sec:non-distinguisher}
We will now prove that some other tempting sampling schemes are not distinguishers.

In the introduction we considered the fastest universal hash function
\texttt{a*x>>63} mapping 64-bit keys to $[2]$. Consider applying it to a pair of keys $(x_1,x_2)$
differ in only in their most significant bit. If the multiplier is odd,
then $ax_1\bmod 2^{64}$ and $ax_2\bmod 2^{64}$ will also only differ in the
most significant bit, so \texttt{a*x>>63} will sample exactly one of
them. On the other hand, if $a$ is even, we will either sample both
$x_1$ and $x_2$, or none of them. Suppose we have one more pair $(y_1,y_2)$,
also differing in only the most significant bit. Then if $a$ is
odd, we sample one key from each pair, and if $a$ is even, we sample
either both or none from each pair. In all cases we end up with an
even number of keys. Thus, if $x_1,x_2,y_1,y_2$ are the keys with
value 1, then the sampled sum over $\Ftwo$ is always $0$.

Now note that \texttt{a*x>>63} is equivalent to \texttt{a*x>t} when
\texttt{t} is the constant $2^{63}$, and  that \texttt{a*x<=t} would
sample the complement set that also has an even number of elements.
The problem persists for schemes like \texttt{a*x+b<=t} as long as
\texttt{t}$=2^{63}$. 
It is thus no coincidence that our \texttt{a*x<=t} distinguisher relies
on \texttt{t} being random.

\section*{Acknowledgments} I would like to thank Martin Dietzfelbinger for 
suggesting some very nice simplifications of the analysis including
the change of summation order used in the proof of Lemma
\ref{lem:good-sum}.  Also, I would like to thank a very thorough
reviewer from FOCS'15 for a 6 page report with many insightful
comments. Finally, I thank Noga Alon for pointing me to the
Paley-Zygmund inequality.

%\bibliographystyle{alpha}
%\bibliography{general}

\end{document}